\documentclass[conference]{IEEEtran}
\makeatletter
\def\ps@headings{%
\def\@oddhead{\mbox{}\scriptsize\rightmark \hfil \thepage}%
\def\@evenhead{\scriptsize\thepage \hfil \leftmark\mbox{}}%
\def\@oddfoot{}%
\def\@evenfoot{}}
\makeatother

\usepackage{amsmath}
\usepackage{paralist}
\usepackage{graphicx}
\usepackage[colorlinks=true]{hyperref} 
\usepackage{color}
\usepackage{arydshln}
\usepackage{cite}
\newcommand{\F}{\mathbf{F}}

\newcommand{\N}{\mathcal{N}}
\newcommand{\G}{\mathcal{G}}

\newtheorem{theorem}{\textbf{Theorem}}
\newtheorem{lemma}[theorem]{\textbf{Lemma}}
\newtheorem{definition}[theorem]{\textbf{Definition}}

\newtheorem{example}[theorem]{Example}

\newcommand{\nix}[1]{}

\begin{document}
\title{\textbf{\emph{SNEED}}: Enhancing  Network Security Services Using  Network Coding and Joint Capacity }
\author{Salah A. Aly~~~~~~~~~Nirwan Ansari~~~~~~~~~~~~ H. Vincent Poor
}

\maketitle
\begin{abstract}
Traditional network security protocols depend mainly on developing cryptographic schemes and on using biometric methods. These have led to several network security protocols that are unbreakable based on difficulty of solving  untractable mathematical problems such as factoring large integers.

In this paper,   \emph{{\underline{S}ecurity of  \underline{N}etworks  \underline{E}mploying \underline{E}ncoding and  \underline{D}ecoding}} (\textbf{SNEED})  is developed to mitigate  single and multiple link attacks. Network coding and shared capacity  among the working paths are used to provide   data protection and  data integrity against network attackers and eavesdroppers.
 \textbf{SNEED} can be incorporated into various applications in  on-demand TV, satellite communications and multimedia security.  Finally, It is shown that \textbf{SNEED} can be implemented easily where there are $k$ edge disjoint paths between two core nodes (routers or switches) in an enterprize network.
\end{abstract}

\section{Introduction}\label{sec:intro}

 Internet Service Providers (ISP) and  Internet Traffic Engineering (ITE) aim to provide fast, reliable, quality of demands, and differentiated services for demanding users. Such services can be deployed at the IP, physical, and application layers.  Several security schemes and network protection strategies have been proposed during the last two decades to protect operational networks against link failures, node attacks, increased overhead and congestion.  The  goal in this paper is to provide novel strategies for network security services against  attacks and eavesdroppers by deploying network coding and shared capacity. The strategies can be deployed to protect network traffics between core nodes such as routers or switches.

Protection of communication networks against network
attacks and failures are  essential to increase robustness, reliability, and availability of the transmitted data.
  The attacks  may also occur at various network layers, including the physical, IP, or application
layers~\cite{stallings06,katz08,menezens01,schneier96}.  Also,  network attacks can occur
due to vulnerable network configurations or  due to
transmitting insecure data.
These problems have received significant
attention from researchers and practitioners, and a large number of techniques have been introduced to
address such problems. In the other side, traditional network security schemes depend
mainly on developing cryptographic protocols or on using biometric
methods. Essentially, cryptographic protocols are
considered unbreakable based on difficulty of solving  mathematical problems such
as factoring large integers~\cite{stallings06,schneier96}.

Network coding is  a powerful tool that has been recently used to increase
the throughput, capacity, and performance of communication networks.
Information theory aspects of network coding have been investigated
in~\cite{soljanin07,fragouli06,ahlswede00} and~\cite{yeung06}. It certainly can offer benefits in terms of
energy efficiency, additional security, and  delay minimization.  Network coding is used to detect adversaries~\cite{ho04} and to  protect packets against network attackers and injectors~\cite{gkantsidis06,jaggi07,cai06}. Network coding can be also used to enhance
security and protection~\cite{jaggi07,lima06}.

 In this paper, we propose an
approach for light-weight network security that is based on network coding. We develop  a scheme called \textbf{SNEED}, Security of   Networks  Employing Encoding and Decoding, in order to protect
transmitted data between  sets of senders and receivers. For one
path that has been attacked (eavesdropped) between a sender and receiver,
 one backup path is provided, in which it will carry encoded data from
sources to  receivers.

\begin{figure}[t]
\begin{center}
  \includegraphics[scale=0.6]{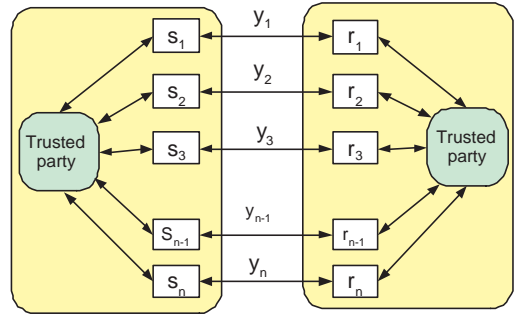}
  \caption{$n$ links are shared between   sets of senders and receivers. One link is used for data integrity and message authentication. The trusted party can be a router (or switch) to send and receive messages.}
  \label{fig:nsenders}
  \end{center}
\end{figure}
\section{Network Model and Assumptions}\label{sec:model}
In this section, we describe the network  model and basic assumptions.  The node can be a router, switch, or an end terminal depending
    on the network model  and  transmission layer.

\subsection{\bf  Network Model}
\begin{compactenum}[i)]
\item Let $\N$ be a network represented by an abstract graph
    $G=(\textbf{V},E)$, where $\textbf{V}$ is a set of nodes and $E$
    is  a set of undirected edges. Let $S =\{s_1,s_2,\ldots \}$ and $R =\{r_1,r_2,\ldots \}$ be  sets of independent sources
    and destinations, respectively. The set $\textbf{V}=V\cup S \cup
    R$ contains the relay nodes, sources, and destinations as shown in
    Fig.~\ref{fig:nsenders}. Assume for simplicity that $|S|=|R|=n$, and  hence the number of sources is equal to the number of receivers. Assume also that the senders (receivers) are connected by a super sender (receiver).

\item A connection path $L_i$ is a set of edges connected together with a
    starting node (sender) and an ending node (receiver).  The  paths $L=\{L_1,L_2,\ldots\}$ carry  data
    from the sources to the receivers. Connection paths are link
    disjoint and provisioned in the network between  senders and
    receivers. All connections have the same bandwidth, or otherwise a
    connection with a high bandwidth can be divided into multiple
    connections, each of which has the unit capacity.

\item Every sender $s_i$ will send a message $m_i^\ell$ to the receiver $r_i$
    at  time $t^\ell_\delta$ in the round $\ell$ in the cycle $\delta$, for all $1 \leq i \leq n$. A receiver $r_i$ receiving a message $m_i^\ell$ is able to detect whether
    the message has been altered by using any authentication or signaling protocols. Further details regarding this model, definition of the working and backup paths, and the normalized capacity can be found in~\cite{aly08i} and~\cite{aly10b}.
    \item $Enc_k$ and $Dec_k$ represent the encryption and decryption algorithms with a shared symmetric key $k$, respectively.
\end{compactenum}

\medskip

\subsection{\bf Senders and Receivers Packets}
Every sender $s_i$ prepares a packet \emph{$packet_{s_i \rightarrow r_i}$}
sent to the receiver $r_i$. The packet contains the sender's $ID_{s_i}$, data $m_{i}^\ell$, and
 time for every round and cycle $t^\ell_\delta$. There are four types of packets that carry the data:

\begin{compactenum}[i)]
\item {\bf Plain Packets.} Packets sent without network coding or encryption, in which the sender does not require to
    perform any coding or encrypting operations.
For example, in case of packets sent without coding, the sender $s_i$
sends the following packet to the receiver $r_i$:
\begin{eqnarray}
packet_{s_i \rightarrow r_i}:=(ID_{s_i},m_i^\ell,t^\ell_\delta)
\end{eqnarray}
\item {\bf Encoded Packets.}  Packets sent with encoded data without encryption,
    in which the sender requires to perform other sender's data. For example, in case of packets sent with encoded data,  the sender $s_i$ sends the
following packet to receiver $r_i$:
\begin{eqnarray}
packet_{s_i \rightarrow r_i}:=(ID_{s_i},\sum_{s_j \in \mathcal{S}} m_j^{\ell},t^\ell_\delta),
\end{eqnarray}
where $\mathcal{S}$ is the set of sources sending plain messages.
\item {\bf Encrypted Packets.}
Assume there is a shared symmetric key between a sender $s_i$ and receiver $r_i$.
In this case the sender $s_i$ will send the packet $packet_{s_i \rightarrow  r_i}$ as follows:
\begin{eqnarray}
packet_{s_i \rightarrow r_i}:=(ID_{s_i}, Enc_{k_i} (m_i^{\ell}),t^\ell_\delta).
\end{eqnarray}
This packet carries an encrypted message without encoding.

\item {\bf Encoded and Encrypted Packets.}  Packets sent with encoded data with encryption,
    in which the sender needs to protect other senders' data. For example, in case of packets sent with encoded encrypted data, the sender $s_i$ sends the
following packet to receiver $r_i$:

\begin{eqnarray}
packet_{s_i \rightarrow r_i}:=(ID_{s_i},\sum_{s_j \in \mathcal{S}} x_j^{\ell},t^\ell_\delta).
\end{eqnarray}
The value $y_i=\sum_{j=1,j\neq i}^n x_j^{\ell}$, where $x_j^\ell=Enc_{k_j} (m_j^\ell)$, is computed by every
sender $s_i$, in which it is able to collect the data from all other
senders and encode them by using the XORed operations.

\end{compactenum}

\medskip

\subsection{\bf   Attackers Model}
We  represent  an attacker model as follows. There are two types of attacks in which the network security services must overcome:  Active (intruders) and passive  (eavesdroppers) attackers~\cite{stinson95}.  We also assume that  $t$ different attackers have access to $t$ channels among all $n$ channels   $L_1,L_2,\ldots,L_n$ at a certain round time, for $t\geq 1$. This is similar to an attacker accessing $t$ channels. Multiple attackers which attack the same channel are represented by one attacker.
\begin{compactenum}[i)]

\item The passive attacker is able to eavesdrop  on the transmission between the senders and receivers. A passive attacker such as an eavesdropper should not learn any information even
if it can have a copy of it. 
\item The active attacker  can modify or fabricate messages throughout a cycle.
    This occurs by injecting new data (coefficients) at the relay nodes of  data sent by the sources. This can occur also over the shared links.  In addition, an active attacker will not be
able to change or fabricate information, and affect the system resources due to the network security
strategies.
\end{compactenum}

This attacker model is similar to the attacker model described for wiretapping channels as stated by many authors~\cite{yeung06,elrouayheb07}.

\section{Data Security Against A Single Attacked Path By Using Shared Keys And Network Coding}\label{sec:SAP}
In this section, we consider the case of a single \emph{active} attacker, i.e., $t=1$.
We will assume that there are shared symmetric keys between the senders and receivers. Also, the receivers are able to detect
messages that have been modified by the attackers using hashing functions like MD5 or SHA-1; fabricated messages
can be detected by using sequence numbers between the senders and receivers~\cite{stallings06,schneier96}.

Let $k_i$ be a shared symmetric key between  $s_i$ and  $r_i$. This key can be distributed by using a Trusted Third Parity (TTP). In this case, the senders exist in a secure domain as well as the receivers.  Let $x_i$ be the encrypted message from the sender $s_i$ to the receiver $r_i$ by using the shared key $k_i$. Thus,
\begin{eqnarray}
x_i^\ell=Enc_{k_i} (m_i^\ell)
\end{eqnarray}

\medskip

\noindent {\bf A. Encoding Operations:} The encoding operations are done as follows. At every round time, $n-1$ senders will send their own data with full capacity over $n-1$ paths that are established from the sources to the destinations. Also, these $n-1$ sources will exchange their data with exactly one  source node $s_i$ that will send the Xored encoded data over a shared link $L_i$.   This process is
explained in Eq.~(\ref{eq:n-1security}) in Table 1.; and we call it \textbf{(SNEED)} against a single attacked path (SAP). The data
is sent in rounds for every cycle. Also, we assume that the attacker can affect only  one path throughout a cycle, but different paths might suffer
from different   \emph{active} attackers throughout different cycles.
\begin{table}\label{table:scheme1}
\caption{The encoding operations of \textbf{SNEED}}
\begin{eqnarray}\label{eq:n-1security}
\begin{array}{|c|ccccccc|c|c|}
\hline
& \multicolumn{7}{|c|}{\mbox{ round time cycle 1 }}&\ldots&\ldots    \\
\hline
&1&2&3&\ldots&j&\ldots&n&\ldots&\ldots   \\
\hline    \hline
  s_1 \rightarrow r_1\!\! &\!\! y_1&x_1^1 &x_1^2&\!\!\ldots\!\!&x_1^j &\!\!\ldots \!\! &\!\! x_1^{n-1}&\ldots&\ldots   \\
    s_2 \rightarrow r_2\!\! &\!\!  x_2^1& y_2& x_2^2&\!\!\ldots\!\!&x_2^j&\!\!\ldots\!\!&\!\! x_2^{n-1}&\ldots&\ldots    \\
s_3 \rightarrow r_3 \!\!& \!\! x_3^1& x_3^2&y_3&\ldots& x_3^j&\!\!\ldots&\!\! x_3^{n-1}&\ldots&\ldots     \\
     \vdots&\vdots&\vdots&\vdots&\vdots&\vdots&\vdots&\vdots&\ldots&\ldots   \\
     s_j \rightarrow r_j\!\! &\!\! x_j^1&x_j^2&x_j^3& \!\!\ldots \!\!&\!\!y_j\!\!&\!\!\ldots\!\!&\!\!x_j^{n-1}&\ldots&\ldots   \\
  \!\!\!\vdots\ddots& \!\!\!\vdots\ddots& \!\!\!\vdots\ddots&\vdots\ddots& \!\!\!\vdots\ddots&\vdots\ddots& \!\!\!\vdots \ddots& \!\!\!\vdots\ddots& \!\!\!\ldots& \!\!\!\ldots   \\
   s_n \rightarrow r_n \!\!&\!\! x_n^1&\!\! x_n^2&\!\! x_n^3&\!\! \ldots\!\!&\!\! x_n^{j-1}&\!\! \ldots\!\!&\!\!\!y_{n}\!\!&\ldots&\ldots   \\
\hline
\hline
\end{array}
\end{eqnarray}
\end{table}
In this case, throughout of one cycle consists of $n$ rounds,  $y_j$'s  for $1\leq j\leq n$ are
defined over $\F_2$ as
\begin{eqnarray} y_j=\sum_{i=1}^{j-1} x_i^{j-1} \oplus \sum_{i=j+1}^n x_i^{j}.
\end{eqnarray}
The senders send packets to the set of receivers in rounds. Every packet
initiated from the sender $s_i$ contains $ID_{s_i}$, data $x_{s_i}^\ell$, and a round
$t_\delta^\ell$. For example, the sender $s_i$ will send the encrypted
$packet_{s_i\longrightarrow r_i}$ as follows.
\begin{eqnarray}
packet_{s_i \longrightarrow r_i}=(ID_{s_i},x_{s_i}^\ell,t_{\delta}^\ell).
\end{eqnarray}
Also, the sender $s_j$ will send the encoded encrypted data $y_{s_j}$ as
\begin{eqnarray}
packet_{s_j \longrightarrow\longrightarrow r_j}=(ID_{s_j},y_{s_j},t_{\delta}^\ell).
\end{eqnarray}
 We ensure that the encoded data $y_{s_j}$ is varied per one round  transmission for every cycle. This means that the path $L_j$ is dedicated
to send only one encoded data $y_j$ and all data
$x_j^1,x_j^2,\ldots,x_j^{n-1}$.

The data transmitted from the sources do not experience any
round time delay. This means that the receivers will be able to decrypt the
received packets online and immediately recover the attacked data. In Eq.~(\ref{eq:n-1security}), the \emph{active} attacker can break only one message per one attacked working path.  A generalization of this scheme is presented in Section~\ref{sec:higherfields} where the \emph{active} attacker(s) has access to $t$ multiple channels simultaneously.

\begin{lemma}
The normalized network capacity according to Eq.~(\ref{eq:n-1security}) is $(n-1)/n$.
\end{lemma}
\begin{proof}
The proof comes from the fact that only one encoded packet is sent over one channel throughout every round per cycle. Therefore, there are $(n-1)$ plain packets sent over $n$ channels.
\end{proof}
\medskip

\noindent {\bf B. Decoding and Data Integrity Operations:} The decoding operations are done as follows. Every receiver $r_i$ will
receive a message $x_i^\ell$ over the link $L_i$. Once the attack occurs at a
link $L_i$, then the receiver that receives $y_\ell$ over the path
$L_\ell$ will be used for data integrity and recovery.

In case the attacker modifies the message, the receiver $r_i$ will know about
the modified message by using MD5 hashing function, so $r_i$ will ask other
receivers to send their messages to recover the modified message.
\begin{eqnarray}
x_i^\ell=y_j\oplus \sum_{h=1,h \neq i}^n x_h^{j-1} \oplus \sum_{h=j+1, h \neq i}^n x_h^j.
\end{eqnarray}
The receiver $r_i$ will decrypt the message $x_i^\ell$ by using its symmetric key $k_i$, i.e., $m_i^\ell=dec_{k_i} (x_i^\ell)$.

\bigskip

\section{\textbf{SNEED} without Sharing symmetric Keys}\label{sec:netsecuritycoding}
In this  section, we propose to use network coding to secure the traffic between the senders and receivers against \emph{active} attackers, where no symmetric keys are shared.  Assume a network with $n$ connections shared between $n$ senders and  receivers.
We will assume that every sender will be able to combine packets from other receivers to hide its own data.
For example, the sender $s_i$ will send the encoded message $y_i$ to the receiver $r_i$ over the link $L_i$.

We will design a security scheme by using network coding against an entity which can not only copy or listen to the message, but also can fabricate new messages or modify the current ones. In this  model we do not assume  pre-shared secret keys between the
senders and receivers. Also, the message  is still secured against attack's fabrication and modification.

\medskip

\begin{figure}[t]
\begin{center}
  \includegraphics[scale=0.65]{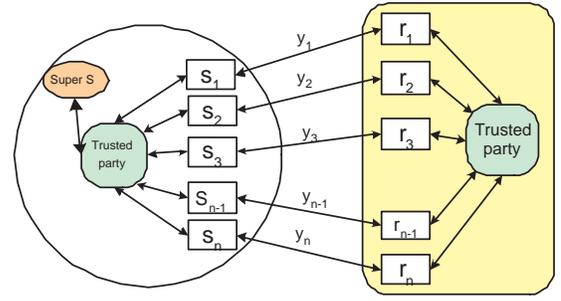}
  \caption{$n$ shared links between  a set of senders and a set of receivers. One link is used for data integrity and authentication. A path is a set of links connecting a set of nodes (routers or switches). A software management to solve a commodity problem can be used to provision a set of link disjoint paths between two core nodes in a given network topology.}
  \label{fig:nsenders2}
  \end{center}
\end{figure}
\subsection{Network Security Codes and  {\bf SNEED}}
We will define  network security codes for {\bf SNEED}  and study their properties. We assume there is a super sender $S$ that sends $n$ different messages over disjoint paths to $n$ receivers as shown in Fig.~\ref{fig:nsenders2}. Furthermore, the receivers can communicate with each other by using a trusted party. The goal is to provide collaborative  security for the $n$ messages against eavesdroppers, intruders, and malicious attackers. We will develop \textbf{SNEED} over the binary field, hence the fastest encoding and decoding operations are used.

One useful application of {\bf SNEED} is the case of sending multimedia and TV streams over a public network as in the Internet. Such streams must be processed online with fast encoding and decoding operations, in addition to the security operations. Let $d$ be the minimum distance defined as in the notion of error correcting codes~\cite{huffman03,macwilliams77,Ayanoglu93}.

\begin{definition}[{\bf SNEED}]\label{def:nsc}
An $[n,k,d]_2$ network security code is a $k$-dimensional subspace of the space $\F_2^n$ that secures $k$ information symbols (messages) by mapping them into $n$ mixed symbols, and can recover from upto $d-1$ compromised (attacked) channels. Furthermore, the code is generated by a nonsystematic matrix $G$ of size $k \times n$ defined over $\F_2$.
\end{definition}

\begin{eqnarray}
G=\begin{pmatrix}
  g_{11} & g_{12} & \ldots & g_{1n} \\
  g_{21} & g_{21} & \ldots & g_{2n} \\
  \vdots & \vdots & \vdots &  \vdots \\
  g_{k1} & g_{k2} & \ldots & g_{kn} \\
\end{pmatrix}_{k \times n}
\end{eqnarray}

We will use the nonsystematic classical binary error correcting codes in the construction of {\bf SNEED}~\cite{huffman03,macwilliams77}.
The encoding scheme of such codes is given by

\medskip

\begin{eqnarray}
\begin{array}{c||cccccccc}
&L_1&L_2&\cdots&L_{n}\\
 \hline \hline
s_1\!\!&g_{11} m_{1}&g_{12} m_{1}&\!\!\ldots&\!\!g_{1n} m_{1}\\
s_2\!\!&g_{21} m_{2}&g_{22} m_{2}&\!\!\ldots&\!\!g_{2n} m_{2}\\
&\vdots&\!\!\vdots&\cdots&\vdots&\\
s_{k}\!\!&g_{k1} m_{k}&g_{k2} m_{k}&\!\!\ldots&\!\!g_{kn} m_{k}\\
\hline \hline \!\!&y_1& y_2&\!\!\ldots&\!\!y_n\\
\end{array}
\end{eqnarray}
The encoding message $y_j$, for $1\leq j \leq n$ ,  is defined by
\begin{eqnarray}
y_j=\sum_{i=1} ^k g_{ij}m_i
\end{eqnarray}
\medskip

\begin{table}[t]
\caption{Best known \textbf{SNEED} codes over $\F_2$~\cite{huffman03}}
\label{table:bchtable}
\begin{center}
\begin{tabular}{|l|l|l|l|}
\hline    n& m &code&type  \\
 \hline
 &&&\\
7&3&$[7,4,3]_2$&Hamming code\\
10&4&$[10,6,3]_2$&Linear code\\
15&4&$[15, 11, 3]_2$& Hamming code\\
19&7&$[19,12,3]_2$&Extension construction\\
23&8&$[23,15,3]_2$& Extension construction \\
25&5&$[25,20,3]_2$&Linear code\\
31&5&$[31, 26, 3]_2$& Hamming code\\
39& 8&$[39, 31,3]_2$&Extension construction \\
47&9&$[ 47,38 ,3]_2$&Extension construction\\
63&6&$[63, 57,3]_2$&Hamming code\\
71& 8&$[71,63,3]_2$&Matrix construction\\
79&9&$[79,70,3]_2$&Extension construction\\
95&10&$[95,85,3]_2$&Extension construction\\
127&7&$[127,120,3]_2$&Hamming code\\
\hline
\end{tabular}
\end{center}
\end{table}

\begin{lemma}
The normalized capacity of the network utilizing {\bf SNEED} is given by
$k/n$.
\end{lemma}
\begin{proof}
By the definition of the network security code, there are $n-k$ redundant symbols that are used to recover from up to $d-1$ attached channels. Therefore there are $k$ working paths that will carry $k$ source data. The result is a consequence by dividing by the total number of channels $n$.
\end{proof}

Table~\ref{table:bchtable} presents the best known \textbf{SNEED} for certain number of channels defined over $\F_2$.

We have $n-k$ lockers' channels, in which they carry redundant data,  in this model. From the proposed code construction, we ensure that $t=d-1\leq n-k$, where $t$ is the number of compromised (attacked) channels. This is actually a direct consequence of the Singleton bound~\cite{huffman03,macwilliams77}.

\subsection{Decoding Operations of  \textbf{SNEED}}
The decoding operations at the receivers side are guaranteed once a system of $t$ linearly independent equations is established in $t$ unknown variables. Let $t$ attackers can access $t$ disjoint channels and alter the transmitted messages. We assumed that the receivers are located in a trusted domain, therefore they can trust and exchange protected messages with  each others.

The system can be solved by using, for example, Gauss elimination method. By definition of \textbf{SNEED}, the matrix $\textbf{G}$ has dimension of $k$. Furthermore, the receivers will know the number and position of the channels that have been attacked. In this case the decoding operations are achieved by using the well known decoding methods for erasure channels, see~\cite{huffman03}.

The following example illustrates the proposed model.
\begin{example}
Assume we have a  connection $L_i$ between a sender  $s_i$ and a receiver  $r_i$ for $i=1,2,3$ and $4$. Furthermore, the channel $L_4$ is used as a lock (redundant) path.
Without loss of generality, we can assume that  the four senders  send
\begin{eqnarray}
\begin{array}{l}
y_1=m_1 \oplus m_2 \\
y_2=m_2 \oplus m_3 \\
y_3=m_1 \oplus m_3 \\
y_4=m_1 \oplus m_2\oplus m_3
\end{array}
\end{eqnarray}
\end{example}
We also assume that the attacker affects only one of channels $L_1,L_2$ or $L_3$.
In this example the \textbf{SNEED} can be stated as follows. There are three working paths and one lock path. The security scheme is given by
\begin{eqnarray}
\begin{pmatrix}
  1 & 0 & 1 & 1 \\
  1 & 1 & 0 & 1 \\
  0 & 1 & 1 & 1 \\
\end{pmatrix} ~~~~~~~ \begin{array}{c||cccc}
&L_1&L_2&L_3 &L_{4}\\
 \hline \hline
s_1\!\!&m_1&0&m_1&m_1\\
s_2\!\!&m_2&m_2&0&m_{2}\\
s_{3}\!\!&0&m_3&m_3&m_3\\
\hline
T&y_1&y_2&y_3&y_4\\
\end{array}
\end{eqnarray}
If the channel $L_2$  is compromised, then the decoding can be done by using Gauss elimination method over the channels $L_1, L_3$ and $L_4$. Adding $y_1$, $y_3$ and $y_4$ will give $m_1$, then substituting in $y_1$ and $y_3$ will give $m_2$ and $m_3$.

\bigskip

\section{\textbf{SNEED} over Higher Finite Fields}\label{sec:higherfields}
In this section, we study security of networks employing encoding and decoding, \textbf{SNEED},  against multiple link attacks.
We   propose \textbf{SNEED} over a finite field with $q$ elements to achieved this goal. In this scheme is an extension of the scheme presented in Section~\ref{sec:SAP}, where the encoding and decoding operations are defined over the binary field. Assume $t$ be the number of compromised channels. One can design a matrix $\G$ over $\F_q$ such that $k=n-t$ paths will carry secure data.
Let $a$ be a primitive element in $\F_q$. The matrix $\G$ is defined by
\begin{eqnarray}
\G=\begin{pmatrix}
  1 & a & a^2 & \ldots & a^{n-1} \\
  1&a^2 & a^4 & \ldots & a^{n-2} \\
  \vdots & \vdots & \vdots &  \vdots&\vdots \\
  1 & a^k & a^{2k}&\ldots & a^{kn} \\
\end{pmatrix}
\end{eqnarray}

The matrix $G$ has rank $k=n-t$. Clearly a finite field with $q> n-t+1$ is sufficient for the encoding and decoding operations. The encoding operations are done by using the following encoding scheme.

\medskip

\begin{eqnarray}
\begin{array}{c||cccccccc}
&L_1&L_2&L_3&\cdots&L_{n}\\
 \hline \hline
s_1\!\!& x_{1}&a x_{1}&a^2 x_1&\!\!\ldots&\!\!a^{1n} x_{1}\\
s_2\!\!& x_{2}&a^2 x_{2}&a^4 x_2&\!\!\ldots&\!\!a^{2n} x_{2}\\
&\vdots&\!\!\vdots&\vdots&\cdots&\vdots&\\
s_{k}\!\!& x_{k}&a^{k} x_{k}&a^{2k}x_k&\!\!\ldots&\!\!a^{kn} x_{k}\\
\hline \hline T\!\!&y_1& y_2&y_3&\!\!\ldots&\!\!y_n\\
\end{array}
\end{eqnarray}

By this construction, if there are up to $t$ attacked paths, then the system of $n-t \times n-t$ equations is solvable. This is due to the fact that the remaining matrix can be reduced to the Vandermond matrix~\cite[Chapter 4]{huffman03}.

\bigskip

\section{Conclusion}\label{sec:conclusion}
Network coding as a promising tool offers benefits for enhancement and supplement of  network security services.
In this paper, we presented schemes for enhancing network security using network coding and joint capacities.  We demonstrated the encoding and decoding operations of the proposed \textbf{SNEED} and showed that it can be deployed over a network with $n$ senders and $n$ receivers. Furthermore,  \textbf{SNEED} is robust against active and passive network attacks. Our future work will include practical aspects of the proposed schemes.

\bigskip

\scriptsize
\bibliographystyle{plain}

\end{document}